\documentclass{lmcs}
\pdfoutput=1

\usepackage{lastpage}
\lmcsdoi{21}{3}{22}
\lmcsheading{}{\pageref{LastPage}}{}{}%
{Aug.~22,~2024}{Aug.~29,~2025}{}


\usepackage{xspace}
\usepackage{amsmath}
\usepackage{amssymb}
\usepackage{stmaryrd}
\usepackage{thm-restate}
\usepackage{tikz}
\usepackage[T1]{fontenc}
\usepackage[utf8]{inputenc}
\usepackage{graphicx}
\usepackage{ifthen}
\usepackage{hyperref}
\usepackage[normalem]{ulem}

\newcommand{\FO}{\textrm{\upshape FO}\xspace}
\newcommand{\SO}{\textrm{\upshape SO}\xspace}
\newcommand{\ML}{\textrm{\upshape ML}\xspace}
\newcommand{\MLD}{\textrm{\upshape ML(D)}\xspace}
\newcommand{\FOxc}{\textrm{\upshape FO$_{\exists,\land}$}\xspace}

\newcommand{\GFO}{\textrm{\upshape GFO}\xspace}
\newcommand{\GNFO}{\textrm{\upshape GNFO}\xspace}
\newcommand{\UNFO}{\textrm{\upshape UNFO}\xspace}
\newcommand{\FOtwo}{\textrm{\upshape FO$^2$}\xspace}
\newcommand{\Ctwo}{\textrm{\upshape C$^2$}\xspace}
\newcommand{\GFOtwo}{\textrm{\upshape GFO$^2$}\xspace}
\newcommand{\FL}{\textrm{\upshape FL}\xspace}
\newcommand{\FF}{\textrm{\upshape FF}\xspace}
\newcommand{\AF}{\textrm{\upshape AF}\xspace}
\newcommand{\GFL}{\textrm{\upshape G$_\FL$}\xspace}
\newcommand{\GFF}{\textrm{\upshape G$_\FF$}\xspace}
\newcommand{\PLTL}{\textrm{\upshape PLTL}\xspace}
\newcommand{\arity}{\textrm{arity}}
\newcommand{\dom}{\textrm{dom}}

\newcommand{\sig}{\textrm{sig}}
\newcommand{\free}{\textrm{free}}
\newcommand{\gfv}{\textrm{gfv}}
\newcommand{\bind}{\textsf{BIND}}
\newcommand{\restrict}{\hspace*{-1.5pt}\mathbin{\upharpoonright}\hspace*{-1.5pt}}

\theoremstyle{plain}

\begin{document}

\title[Craig Interpolation for Decidable First-Order Fragments]{Craig Interpolation for Decidable First-Order Fragments}
\titlecomment{{\lsuper*}This is an extended version of a paper published at FoSSaCS 2024 \cite{tenCate2024craig}.}

\author[B.~ten Cate]{Balder ten Cate\lmcsorcid{0000-0002-2538-5846}}[a]
\author[J.~Comer]{Jesse Comer\lmcsorcid{0009-0006-9734-3457}}[b]

\address{ILLC, University of Amsterdam, Amsterdam 1098 XH, NL}
\email{b.d.tencate@uva.nl}  
\address{University of Pennsylvania, Philadelphia, PA 19104, USA}
\email{jacomer@seas.upenn.edu}  



\begin{abstract}
We show that the guarded-negation fragment is, in a precise sense, the smallest extension of the guarded fragment with Craig interpolation. In contrast, we show that full first-order logic is the smallest extension of both the two-variable fragment and the forward fragment with Craig interpolation. Similarly, we also show that all extensions of the two-variable fragment and of the fluted fragment with Craig interpolation are undecidable.
\end{abstract}

\maketitle

\section{Introduction}
\label{sec:intro}
The study of decidable fragments of first-order logic (\FO) is a topic with a long history, dating back to the early 1900s (\cite{Lowenheim1915,Skolem1920}, see~also~\cite{Borger1997:classic}), and more actively pursued since the 1990s. Inspired by Vardi~\cite{Vardi1996:why}, who asked ``what makes modal logic so robustly decidable?'' and Andreka et al.~\cite{Andreka1998:Modal}, who asked ``what makes modal logic tick?'' many decidable fragments have been introduced and studied over the last 25 years that take inspiration from modal logic (\ML), which itself can be viewed as a fragment of \FO that features a restricted form of quantification. These include the following fragments, each of which naturally generalizes modal logic in a different way: the \emph{two-variable fragment} (\FOtwo)~\cite{henkin1967logical, Mortimer1975:languages}, the \emph{guarded fragment} (\GFO)~\cite{Andreka1998:Modal}, and the \emph{unary negation fragment} (\UNFO)~\cite{tencate2013:unary}. Further decidable extensions of these fragments were subsequently identified, including the \emph{two-variable fragment with counting quantifiers} (\Ctwo)~\cite{Graedel97:two,pacholski1997complexity} and the \emph{guarded negation fragment} (\GNFO)~\cite{Barany2015:guarded}. The latter can be viewed as a common generalization of \GFO and \UNFO. Many decidable logics used in computer science and AI, including various description logics and rule-based languages, can be translated into \GNFO and/or \Ctwo. In this sense, \GNFO and \Ctwo are convenient tools for explaining the decidability of other logics. Extensions of \GNFO have been studied that push the decidability frontier even further (for instance with fixed-point operators \cite{benedikt2016step, Benedikt2019:definability} and using clique-guards), but these fall outside the scope of this paper.

In an earlier line of investigation, Quine identified the decidable \emph{fluted fragment} (\FL) \cite{quine1969}, the first of several \emph{ordered logics} which have been the subject of recent interest \cite{purdy1996-1,purdy1996-2,purdy1999,purdy2002,pratt-hartmann2019,bednarczyk2022,bednarczyk2023adjacent}. The idea behind ordered logics is to restrict the order in which variables are allowed to occur in atomic formulas and quantifiers. Another recently introduced decidable fragment that falls in this family is the \emph{forward fragment} (\FF), whose syntax strictly generalizes that of \FL. Both \FL and \FF have the finite model property 
(FMP) \cite[Theorem 4.10]{pratt-hartmann2019} and embed \ML \cite{hustadt2004}, but are incomparable in expressive power to \GFO \cite{pratt-hartmann2016}, \FOtwo, and \UNFO.\footnote{Specifically, the \FO-sentence $\exists x y (R(x,y) \land R(y,x))$ belongs to \GFO, \FOtwo and \UNFO, but is not expressible in \FF, since the structure consisting of two points with symmetric edges and the structure $(\mathbb{Z},S)$ with $S$ the successor relation, are ``infix bisimilar,'' as described in \cite{bednarczyk2022}.} The recently-introduced \emph{adjacent fragment} (\AF) generalizes both \FF and \FL while also retaining decidability~\cite{bednarczyk2023adjacent}.

Ideally, an \FO-fragment is not only decidable, but also model-theoretically well behaved. A particularly important model-theoretic property of logics is the \emph{Craig Interpolation Property} (CIP). It states that, for all formulas $\varphi, \psi$, if $\varphi \models \psi$, then there exists a formula $\vartheta$ such that $\varphi \models \vartheta$ and $\vartheta \models \psi$, and such that all non-logical symbols occurring in $\vartheta$ occur both in $\varphi$ and in $\psi$. Craig~\cite{craig1957linear,Craig1957three} proved in 1957 that \FO itself has this property (hence the name). Several refinements of Craig's result have subsequently been obtained (e.g.,~\cite{Otto2000:interpolation,benedikt2016generating}). These have found  applications in various areas of computer science and AI, including formal verification, modular hard/software specification and automated deduction~\cite{DBLP:reference/mc/McMillan18,CalEtAl20,DBLP:conf/aplas/HoderHKV12}, and are emerging as a new prominent technology in databases~\cite{DBLP:series/synthesis/2011Toman,benedikt2016generating} and knowledge representation~\cite{DBLP:conf/ijcai/LutzW11,tencate2013:beth,DBLP:conf/aaai/KoopmannS15}. While we have described CIP here as a model-theoretic property, it also has a proof-theoretic interpretation. Indeed, it has been argued that CIP is an indicator for the existence of nice proof systems~\cite{hooglandthesis}.

Turning our attention to the decidable fragments of \FO we mentioned earlier, it turns out that, although \GFO is in many ways model-theoretically well-behaved~\cite{Andreka1998:Modal}, it lacks CIP~\cite[Theorem~3.2]{Hoogland02:interpolation}. Likewise, \FOtwo lacks CIP~\cite[Corollary 4]{Comer1969:classes} and the same holds for \Ctwo  (\cite[Example~2]{Jung2021:living} yields a counterexample). Both \FF and \FL lack CIP~\cite[Theorem~9]{bednarczyk2022}. On the other hand, CIP is enjoyed by both \UNFO \cite[Theorem~3.10]{tencate2013:unary} and \GNFO \cite[Theorem~7]{Benedikt2013:rewriting}. Figure~\ref{fig:fragments} summarizes these known results. Note that we restrict attention to relational signatures without constant symbols and function symbols. Some of the results depend on this restriction. Other known results not reflected in Figure~\ref{fig:fragments} (to avoid clutter) are that the intersection of \GFO and \FOtwo (also known as \GFOtwo) has CIP~\cite[Theorem~6.1]{Hoogland02:interpolation}. Similarly, the intersection of \FF with \GFO and the intersection of \FL with \GFO (known as \GFF and \GFL, respectively) have CIP~\cite[Theorem~13]{bednarczyk2022}.

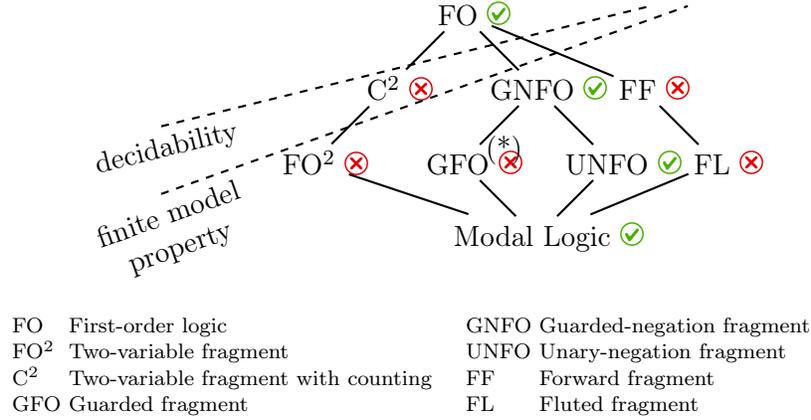
\begin{figure}[t]
\centering
\vspace{-2mm}
\newcommand{\yes}[1][\!\!\!\!\!]{\includegraphics[scale=.02]{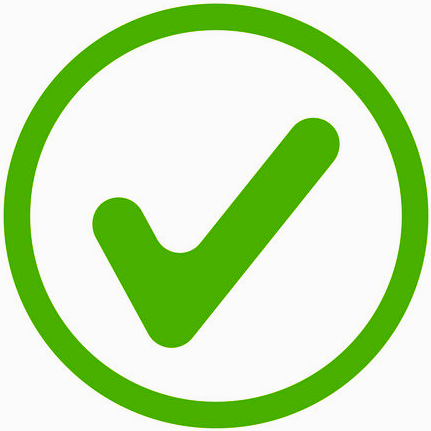}#1}
\newcommand{\no}[1][\!\!\!\!\!]{\includegraphics[scale=.02]{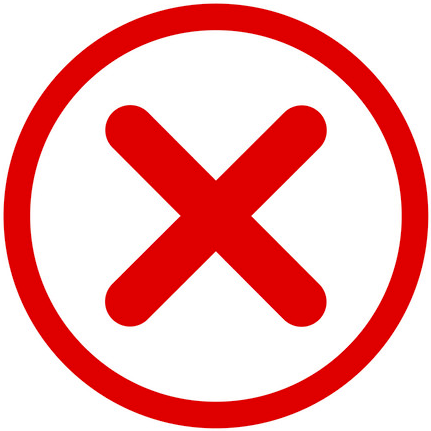}#1}
\begin{tikzpicture}
  [align=center,node distance=1.4cm, every node/.style={scale=1}]
    \node (fo) at (0,0) {\FO \yes};
    \node [below left of=fo] (ctwo) {\Ctwo \no};
    \node [below right of=fo] (gnfo) {\GNFO \yes};
    \node [below left  of=gnfo] (gfo)  {\GFO \no};
    \node [below right of=gnfo] (unfo) {\UNFO \yes};
    \node [below left of=ctwo] (twovar) {\FOtwo \no}; 
    \node [below right of= gfo] (ml) {Modal Logic \yes};
    \node [right of=gnfo] (ff) {FF \no}; 
    \node [right of=unfo] (fl) {FL \no}; 
    
    \draw [thick, shorten <=-2pt, shorten >=-2pt] (fo) -- (gnfo);
    \draw [thick, shorten <=-2pt, shorten >=-2pt] (fo) -- (ctwo);
    \draw [thick, shorten <=-2pt, shorten >=-2pt] (ctwo) -- (twovar);
    \draw [thick, shorten <=-2pt, shorten >=-2pt] (gnfo) -- node[below] {~~~(*)} (gfo);
    \draw [thick, shorten <=-2pt, shorten >=-2pt] (gnfo) -- (unfo);
    \draw [thick, shorten <=-2pt, shorten >=-2pt] (gfo) -- (ml);
    \draw [thick, shorten <=-2pt, shorten >=-2pt] (twovar) -- (ml);
    \draw [thick, shorten <=-2pt, shorten >=-2pt] (unfo) -- (ml);
      \draw [thick, shorten <=-2pt, shorten >=-2pt] (fo) -- (ff);
      \draw [thick, shorten <=-2pt, shorten >=-2pt] (ff) -- (fl);
      \draw [thick, shorten <=-2pt, shorten >=-2pt] (fl) -- (ml);
  
    \draw [dashed, thick] (-4,-1.5) -- 
      node[below, at start, sloped] {decidability} (2.5,0.1);
    \draw [dashed, thick] (-4,-2.4) -- node[below, at start, sloped] {finite model\\ property} (3,0.1);
\end{tikzpicture}

{
\scriptsize
\begin{tabular}[t]{l@{~~}l}
\FO & First-order logic\\
\FOtwo & Two-variable fragment \\
\Ctwo & Two-variable fragment  with counting \\
\GFO & Guarded fragment \\
\end{tabular} ~~
\begin{tabular}[t]{l@{~~}l}
\GNFO & Guarded-negation fragment \\
\UNFO & Unary-negation fragment \\
\FF & Forward fragment \\
\FL & Fluted fragment 
\end{tabular}
}

\caption{Landscape of decidable fragments of \FO with (\yes[]) and without (\no[]) CIP. \\ The inclusion marked $(*)$ holds only for sentences and self-guarded formulas.}
\label{fig:fragments}
\end{figure}

When a logic $L$ lacks CIP, the question naturally arises as to whether there exists a more expressive logic $L'$ that has CIP. If such an $L'$ exists, then, in particular, interpolants for valid $L$-implications can be found in $L'$. This line of analysis is sometimes referred to as \emph{Repairing Interpolation} \cite{Areces03:repairing}. If $L'$ is an \FO-fragment, and our aim is to repair interpolation by extension, then there is a trivial solution: \FO itself is an extension of $L$ satisfying CIP. We will instead consider the following refinement of the question: can a natural extension $L'$ of $L$ be identified which satisfies CIP while retaining decidability? We will answer this question for three of the fragments depicted in Figure~\ref{fig:fragments} that lack CIP, by identifying the minimal natural extension $L'$ of $L$ satisfying CIP. Our main results can be  stated informally as follows:
\begin{enumerate}
    \item The smallest logic extending \GFO that has CIP is \GNFO.
    \item The smallest logic extending \FOtwo that has CIP is \FO.
    \item The smallest logic extending \FF that has CIP is \FO.
    \item No decidable extensions of \FOtwo or \FL have CIP.
\end{enumerate}

\noindent 
The precise statements of these results will be given in the respective sections. They involve some natural closure assumptions on the logics in question, and, for the undecidability results, some assumptions regarding the effective computability of the translation between the extension and the logic that it extends.

These results give us a clear sense of where, in the larger landscape of decidable fragments of \FO, we may find logics that enjoy CIP. What makes the above results remarkable is that, from the definition of the Craig interpolation property, it does not appear to follow that a logic without CIP would have a unique minimal extension with CIP. Note that a valid implication may have many possible interpolants, and the Craig interpolation property merely requires the existence of one such interpolant.  Nevertheless, the above results show that, in the case \FOtwo, \GFO, and \FF, such a unique minimal extension indeed exists (assuming suitable closure properties, which will be spelled out in detail in the next sections).

\subsubsection*{Related Work.}
Several other approaches have been proposed for dealing with logics that lack CIP. One approach is to weaken CIP. For example, it was shown in \cite[Theorem~4.5]{Hoogland02:interpolation} that \GFO satisfies a weak, ``modal'' form of Craig interpolation, where, roughly speaking, only the relation symbols that occur in non-guard positions in the interpolant are required to occur both in the premise and the conclusion. As it turns out, this weakening of CIP is strong enough to entail the (non-projective) \emph{Beth Definability Property}, which is one important use case of CIP. See also Section~\ref{sec:conclusion} for further discussion of  weak forms of CIP. 

Another recent approach~\cite{Jung2021:living} is to develop algorithms for testing whether an interpolant exists for a given entailment. That is, rather than viewing Craig interpolation as a property of logics, the existence of interpolants is studied as an algorithmic problem at the level of individual entailments. The interpolant existence problem turns out to be indeed decidable (with higher complexity than the satisfiability problem) for both \GFO and \FOtwo~\cite[Theorem~1]{Jung2021:living}, although it is undecidable for \FOtwo with two equivalence relations~\cite[Theorem~5]{wolter2024interpolation}.

Additional results are known for \UNFO and \GNFO beyond the fact that they have CIP. In particular, CIP holds for the fixed-point extension of \UNFO, while the weak ``modal'' form of Craig interpolation, mentioned above, holds for the fixed-point extension of \GNFO~\cite{Benedikt2015:interpolation,Benedikt2019:definability}. Furthermore, interpolants for \UNFO and \GNFO can be constructed effectively, and tight bounds are known on the size of interpolants and the computational complexity of computing them~\cite{Benedikt2015:effective}.

Our paper can be viewed as an instance of \emph{abstract model theory} for fragments of \FO. One large driving force behind the development of abstract model theory was the identification of \emph{extensions} of \FO which satisfy desirable model-theoretic properties, such as compactness, L\"owenheim-Skolem, and Craig interpolation. One takeaway from this line of research is that CIP is scarce among many ``reasonable'' \FO-extensions. An early result of Lindstr\"om showed that \FO-extensions with finitely many generalized quantifiers and satisfying the downward L\"owenheim-Skolem property do not have the Beth property (and hence fail to satisfy CIP) \cite[Theorem~5]{lindstrom1969}. Similarly, Caicedo~\cite[Theorem~2.2]{caicedo1985}, generalizing an early result by Friedman \cite{friedman1973}, established a strong negative CIP result that applies to arbitrary proper \FO-extensions with monadic generalized quantifiers. For a survey of negative interpolation results among \FO-extensions, see \cite{vaananen2008}. These negative results not only show that CIP is scarce among extensions of \FO,  they also provide clues as to where, within the space of all extensions, one may hope to find logics with CIP. Our results can be viewed similarly, except that they pertain to (extensions of) fragments of \FO.

Our results can also be appreciated as characterizations of \GNFO and of \FO. While traditional Lindstr\"om-style characterizations are maximality theorems (e.g., \FO is a maximal logic having the compactness and L\"owenheim-Skolem properties), our results can be viewed as minimality theorems (e.g., \GNFO is the minimal logic extending \GFO and having CIP). 

Some prior work exists that studies abstract model theory  for (extensions of) fragments of \FO. Most closely related is \cite{tencate2005:interpolation}, which studies modal logics and hybrid logics. Among other things, it was shown in~\cite[Theorem~4.1]{tencate2005:interpolation} that the smallest extension of modal logic with the difference operator (\MLD) which satisfies CIP is full first-order logic. Additionally, in~\cite{gheerbrant2009craig}, the authors identified minimal extensions of various fragments of propositional linear temporal logic (\PLTL) with CIP. Furthermore, it was shown in~\cite[Corollary~5.1]{tencate2005:interpolation} that every abstract logic extending \GFO with CIP can express all \FO sentences and formulas with one free variable, and is thus undecidable. A crucial difference between this result and ours is that \cite{tencate2005:interpolation} assumes signatures with constant symbols and concerns a stronger version of CIP, interpolating not only over relation symbols but also over constant symbols. In contrast, we consider here only purely relational signatures without constant symbols. Other prior work on abstract model theory for  fragments of \FO are \cite{vanBenthem2007:lindstrom,vanbenthem2009:lindstrom,garciamatos2005}. Repairing interpolation has also been pursued in the context of quantified modal logics, which typically lack CIP; in \cite[Theorem~4]{Areces03:repairing}, the authors showed that CIP can be repaired for such logics by adding nominals, @-operators and the $\downarrow$-binder.

\subsubsection*{Outline.}
Section \ref{sec:prelim} introduces the abstract model-theoretic framework. In Sections \ref{sec:rep-fo2}, \ref{sec:rep-gfo}, and \ref{sec:rep-ff}, we repair interpolation for \FOtwo, \GFO, and \FF, respectively. In Section \ref{sec:weak-ext}, we provide results showing that, even with weak expressive assumptions, extensions of \FOtwo and \FL with CIP are undecidable. In Section \ref{sec:conclusion}, we discuss the implications and limitations of our results, and future directions.

\section{Preliminaries}
\label{sec:prelim}
We assume familiarity with structures, reducts, and expansions, as well as the syntax and semantics of \FO~\cite{enderton2001mathematical}. Signatures are denoted by $\sigma$ and $\tau$, and are assumed to be relational and finite. If $\varphi$ contains only relation symbols occurring in $\sigma$, then we write $M, g \models \varphi$ to denote that a $\sigma$-structure $M$ satisfies $\varphi$ under the variable assignment $g$. We write $x_i, y_i, z_i, u_i$ to denote variables, and $\overline{x}, \overline{y}, \overline{z}, \overline{u}$ to denote tuples of variables. We write $a_i,b_i,c_i$ to denote elements of structures and $\overline{a},\overline{b},\overline{c}$ to denote tuples of such elements. Given a tuple of elements $\overline{a} = a_1,\hdots,a_n$ in a structure $M$, a tuple of variables $\overline{x} = x_1,\hdots,x_n$, and a variable assignment $g$, we write $g[\overline{x} / \overline{a}]$ to denote the variable assignment which is the same as $g$ except that $g(x_i) = a_i$ for each $i \leq n$. In order to state our main results precisely, we must formally define what we mean by \emph{extensions $L'$ of $L$} (where $L$ is some fragment of \FO that lacks CIP). One option is to let $L'$ range over fragments of FO that syntactically include $L$. However, as it turns out, our main results apply even to extensions that are not themselves contained in \FO. We therefore opt, instead, to work with an abstract definition of logics, as typically used in abstract model theory.

\subsubsection*{Abstract Logics.}
An \emph{abstract logic} (or, simply, a \emph{logic}) is a pair $(L,\models_L)$, where $L$ is a map from relational signatures $\sigma$ to collections of \emph{formulas}, and $\models_L$ is a ternary \emph{satisfaction relation}. A \emph{formula} of an abstract logic $(L,\models_L)$ is an element of $L(\sigma)$ for some finite relational signature $\sigma$. $L$ must be monotone: if $\sigma \subseteq \tau$, then $L(\sigma) \subseteq L(\tau)$. Each formula $\varphi$ has an associated finite set of free variables $\free(\varphi)$, and we write $\varphi(\overline{x})$ or $\varphi(x_1,\dots,x_k)$ to denote that the free variables of $\varphi$ are exactly those in the tuple $\overline{x} = x_1,\hdots,x_k$. As in the case of \FO, a formula $\varphi$ is a \emph{sentence} if $\free(\varphi) = \emptyset$. We write $\sig(\varphi)$ to denote the least signature $\sigma$ such that $\varphi \in L(\sigma)$. The ternary \emph{satisfaction relation} $\models_L$ is defined over triples $(M,g,\varphi)$, where $\varphi$ is an $L$-formula, $M$ is a $\tau$-structure such that $\sig(\varphi) \subseteq \tau$, and $g$ is a variable assignment with $\free(\varphi) \subseteq \dom(g)$; we write $M, g \models_L \varphi$ if this relation holds between these objects. The notions of logical consequence and logical equivalence for abstract logics are defined completely analogously to \FO. In later sections, we will prefer to suppress the subscript $L$ in the notation for the satisfaction relation and write $L$ to denote an abstract logic $(L,\models_L)$. Furthermore, we often write $\varphi \in L$ rather than $\varphi \in L(\sigma)$, leaving the signature implicit. \looseness=-1

All abstract logics $L$ are assumed to satisfy the \emph{reduct property} and the \emph{renaming property}. The \emph{reduct property} states that if $\sigma \subseteq \tau$, then for all $\varphi \in L(\sigma)$, all $\tau$-structures $M$, and all assignments $g$, if $M,g \models_L \varphi$, then $M \restrict \sigma, g \models_L \varphi$, where $M \restrict \sigma$ denotes the $\sigma$-reduct of $M$. In other words, the truth of a formula of an abstract logic $L$ in a structure depends only on the interpretations of the symbols in the signature of that formula. The \emph{renaming property} states that if $\rho: \sigma \to \tau$ is an injective map preserving the arity of relation symbols, then for each formula $\varphi \in L(\sigma)$, there is a formula $\psi \in L(\tau)$ such that for all $\tau$-structures $M$, we have that $M, g \models_L \psi$ if and only if $\rho^{-1}[M], g \models_L \varphi$, where $\rho^{-1}[M]$ is the $\sigma$-structure with the same domain as $M$ where, for each $R \in \sigma$, we have that $R^{\rho^{-1}[M]} = \rho(R)^M$. Intuitively, the renaming property states that if a formula over a signature $\sigma$ can be expressed in a logic $L$, then the formula obtained by renaming all of its relation symbols can also be expressed in $L$. \looseness=-1

For an arbitrary abstract logic $L$, the Craig interpolation property states that if $\varphi \models_L \psi$ for $L$-formulas $\varphi$ and $\psi$, then there exists an $L(\sig(\varphi) \cap \sig(\psi))$-formula $\vartheta$ such that $\varphi \models_L \vartheta$ and $\vartheta \models_L \psi$, where $\free(\vartheta) = \free(\varphi) \cap \free(\psi)$. 

We say a formula $\varphi$ of a logic $L$ expresses a formula $\psi$ of a logic $L'$ if $\free(\varphi) = \free(\psi)$, $\sig(\varphi) = \sig(\psi)$, and for all structures $M$ and assignments $g$, we have that $M, g \models_L \varphi$ if and only if $M, g \models_{L'} \psi$. We say that a logic $L'$ is an \emph{extension} of a logic $L$ (notation: $L \preceq L'$) if $L'$ can express all formulas of $L$. An \FO-fragment can then be precisely defined, without reference to syntax, as a logic of which \FO is an extension. We say that $L'$ is a \emph{sentential extension} of $L$ (notation: $L \preceq_{sent} L'$) if $L'$ can express all sentences of $L$.

Let $L$ be a logic and $\psi(x_1,\hdots,x_n)$ be an $L$-formula. We write $\llbracket \psi \rrbracket^M$ for the collection of tuples $(a_1,\hdots,a_n) \in M^n$ such that there exists an assignment $g$ where $M, g \models \psi$ and $g(x_i) = a_i$ for each $i \leq n$. Given formulas $\psi_1,\hdots,\psi_k \in L(\sigma)$, a $\sigma$-structure $M$, and relation symbols $R_1,\hdots,R_k \in \sigma$ with $\lvert \free(\psi_i) \rvert = \arity(R_i)$ for each $i \leq k$, we define $M[R_1/\psi_1,\hdots,R_k/\psi_k]$ to be the $\sigma$-structure with the same domain as $M$ and such that $R_i^{M[R_1/\psi_1,\hdots,R_k/\psi_k]} = \llbracket \psi_i \rrbracket^M$ for each $i \leq k$. We now describe a syntax-free notion of uniform substitution for formulas of an abstract logic.

\begin{defi}
\label{def:abstract-sub}
Let $L$ be a logic and $\varphi \in L(\sigma)$ with $R_1,\hdots,R_k \in \sig(\varphi)$, where for each $i \leq k$, we have that $R_i$ is a $k_i$-ary relation symbol. Furthermore, let $\psi_1,\hdots,\psi_k \in L(\sigma)$ be formulas with $\lvert \free(\psi_i) \rvert = k_i$ for each $i \leq k$. We say that $L$ \emph{expresses the substitution of $\psi_1,\hdots,\psi_k$ for $R_1,\hdots,R_k$ in $\varphi$} if there exists a formula $\chi \in L(\sigma)$ such that, for every $\sigma$-structure $M$,
\[ M, g \models \chi \iff M[R_1/\psi_1,\hdots,R_k/\psi_k], g \models \varphi. \]
\end{defi}

Most studies in abstract logic assume that the logics under study are \emph{regular}, roughly meaning that they can express atomic formulas, Boolean connectives, and existential quantification. In other words, to study regular logics is to study extensions of \FO. Since we are interested in a more fine-grained view of logics including \FO-fragments, these assumptions are too strong. As a result, the first step of studying extensions of \FO-fragments from the perspective of abstract logic is to identify natural expressive assumptions for those extensions which are strictly weaker than regularity. We do this in the respective sections.

Some of our proofs will use second-order quantification (for expository reasons only), and we recall the semantics of these quantifiers here. Given a formula $\varphi \in L(\sigma \cup \{P\})$ of some abstract logic $L$, we can form new formulas $\exists P \varphi$ and $\forall P \varphi$ with signature $\sigma$ and the same free variables as $\varphi$. Given a $\sigma$-structure $M$ and an assignment $g$, the semantics of these formulas are defined as follows:
\begin{align*}
M,g \models \exists P \varphi ~~&\text{if there is a}~ (\sigma \cup\{P\})\text{-expansion}~ M' ~\text{of}~ M \\
&\text{such that}~ M', g \models \varphi, ~\text{and} \\
M,g \models \forall P \varphi ~~&\text{if for all}~ (\sigma \cup\{P\})\text{-expansions}~ M' ~\text{of}~ M, \\
&\text{we have that}~ M', g \models \varphi.
\end{align*}
If $L$ itself does not allow second-order quantification, we can view $\exists P \varphi$ and $\forall P \varphi$ as elements of $L'(\sigma)$ for a suitable extension $L'$ of $L$. In particular, if $\varphi$ is an \FO-formula, then $\exists P \varphi$ and $\forall P \varphi$ are formulas of second-order logic (\SO).

\section{Repairing Interpolation for \texorpdfstring{\FOtwo}{the Two-Variable Fragment}}
\label{sec:rep-fo2}
The two-variable fragment (\FOtwo) consists of all FO-formulas containing only two variables, say, $x$ and $y$, where we allow for nested quantifiers that reuse the same variable (as in $\exists xy (R(x,y)\land \exists x(R(y,x)))$, expressing the existence of a path of length 2). In this context, as is customary, we restrict attention to relations of arity at most $2$. It is known that \FOtwo is decidable~\cite{Mortimer1975:languages} but does not have CIP~\cite{Comer1969:classes}. \looseness=-1

\subsection{Natural Extensions of \texorpdfstring{\FOtwo}{the Two-Variable Fragment}}

While \FOtwo is restricted to only two variables and predicates of arity as most 2, it has no restriction on its connectives: it is fully closed under Boolean connectives and existential and universal quantification. Because of this fact, we will consider in this section those abstract logics which are \emph{strong extensions} of \FOtwo.

\begin{defi}
\label{def:strong-ext}
We say that a logic $L'$ \emph{strongly extends} a logic $L$ if $L'$ extends $L$ and, for each formula $\varphi \in L'$ with $R_1,\hdots,R_k \in \sig(\varphi)$, where $\varphi$ expresses some $\psi \in L$, and all formulas $\psi_1,\hdots,\psi_k \in L'$, we have that $L'$ expresses the substitution of $\psi_1,\hdots,\psi_k$ for $R_1,\hdots,R_k$ in $\varphi$ (cf. Definition \ref{def:abstract-sub}).
\end{defi}

Intuitively, Definition \ref{def:strong-ext} means that $L'$ can express uniform substitutions of its formulas into formulas of $L$. In other words, the notion of a strong extension is a syntax-free way to say that $L'$ extends $L$ and is closed under the connectives of $L$. In particular, if $L$ strongly extends \FOtwo, then $L$ can express all of the usual first-order connectives: if $\psi_0$ and $\psi_1$ are expressible in $L$, then $\lnot \psi_0$, $\psi_0 \land \psi_1$, and $\exists x \psi_0$ are also expressible in $L$, under the usual semantics of these connectives. Clearly \Ctwo is a strong extension of \FOtwo, and \FOtwo is the smallest strong extension of itself.

\subsection{Finding the Minimal Extension of \texorpdfstring{\FOtwo}{the Two-Variable Fragment} with CIP}

Recall that we write $L \preceq_{sent} L'$ if every sentence of $L$ is expressible in $L'$. Our main result in this section is the following.

\begin{thm}
\label{thm:FOtwo-main}
If $L$ is a strong extension of \FOtwo with CIP, then $\FO \preceq_{sent} L$.
\end{thm}
\begin{proof}
We will show by induction on the complexity of formulas that, for every $\FO$-formula $\varphi(x_1 \ldots, x_n)$ there is a sentence $\psi\in L$ over an extended signature containing additional unary predicates $P_1, \ldots, P_n$, that is equivalent to 
\[ \exists x_1\ldots x_n(\big(\!\!\!\bigwedge_{i=1\ldots n}\!\!\! P_i(x_i)\land\forall y(P_i(y)\to y=x_i)\big)\land\varphi(x_1, \ldots, x_n)). \]

In other words, $\psi$ is a sentence expressing that $\varphi$ holds under an assignment of its free variables to some tuple of elements which uniquely satisfy the $P_i$ predicates. In the case that $n=0$ (i.e., the case that $\varphi$ is a sentence), we then have that $\psi$ is equivalent to $\varphi$, which shows that $FO\preceq_{sent} L$.

The base case of the induction follows from the fact that we restrict attention to relations of arity at most 2. The induction step for the Boolean connectives is straightforward as well (using the fact that $L$ is a strong extension of \FOtwo, and thus can express all connectives of \FOtwo). In fact, the only non-trivial part of the argument is the induction step for the existential quantifier. Let $\varphi(x_1,\ldots,x_n)$ be of the form $\exists x_{n+1} \varphi'(x_1.\ldots,x_n,x_{n+1})$. By the inductive hypothesis, there is an $L$-sentence $\psi$ with $\sig(\psi) = \sig(\varphi') \cup \{P_1,\hdots,P_{n+1}\}$, where $P_1,\hdots,P_{n+1}$ are unary predicates not in $\sig(\varphi')$, which is equivalent to
\[ \exists x_1\ldots x_n\exists x_{n+1}(\big(\!\!\!\bigwedge_{i \leq n+1}\!\!\! P_i(x_i)\land\forall y(P_i(y)\to y=x_i)\big)\land\varphi'(x_1, \ldots, x_n,x_{n+1})). \] 
Now, let $\psi'$ be obtained from $\psi$ by replacing every occurrence of $P_{n+1}$ by $P'$ for some fresh unary predicate $P'$; this is expressible in $L$ by the renaming property. Furthermore, let
\[ \begin{array}{ll} 
\gamma(x) &:= \psi \land P_{n+1}(x), ~\text{and}\\[2mm]
\chi(x)   &:= (P'(x)\land\forall y(P'(y)\to y=x)) \to \psi'.
\end{array} \]
(where $x$ is either of the two variables we have at our disposal; it does not matter which). Since $L$ strongly extends \FOtwo, both can be written as an $L$-formula. Then
\[ \gamma(x) \models \chi(x). \]
Let $\vartheta(x)\in L$ be an interpolant. Observe that since $P_{n+1}$  occurs only in $\gamma(x)$ and $P'$ only in $\chi(x)$, the following second-order entailment is also valid:
\[ \exists P_{n+1}\gamma(x) \models \vartheta(x) \models \forall P' \chi(x). \]
It is not hard to see that $\exists P_{n+1} \gamma(x)$ and $\forall P' \chi(x)$ are equivalent. Indeed, both are satisfied in a structure $M$ under an assignment $g$ precisely if $M',g\models\varphi$, where $M'$ is the expansion of $M$ in which $P_{n+1}$ denotes the singleton set $\{g(x_{n+1})\}$. It then follows that $\vartheta(x)$, being sandwiched between the two, is also equivalent to $\exists P_{n+1} \gamma(x)$. This implies that $\vartheta(x)$ is the unique interpolant (up to logical equivalence) of the entailment $\gamma(x) \models \chi(x)$, and so it is expressible in $L$. Then since $L$ strongly extends \FOtwo, it can express $\exists x \vartheta(x)$. We claim that this sentence satisfies the requirement of our claim. To see this, observe that $\exists x \vartheta(x)$ is equivalent to $\exists x \exists P_{n+1} \gamma(x)$, which is equivalent to $\exists P_{n+1} \psi$, which clearly satisfies the requirement of our claim.
\end{proof}

\section{Repairing Interpolation for \texorpdfstring{\GFO}{the Guarded Fragment}}
\label{sec:rep-gfo}
The guarded fragment (\GFO) \cite{Andreka1998:Modal} allows formulas in which all quantifiers are ``guarded.'' Formally, a \emph{guard} for a formula $\varphi$ is an atomic formula $\alpha$ whose free variables include all free variables of $\varphi$. Following~\cite{Graedel99:restraining}, we allow $\alpha$ to be an equality. More generally, by an \emph{$\exists$-guard} for $\varphi$, we mean a formula of the form $\exists \overline{x} \beta$ whose free variables include all free variables of $\varphi$, where $\beta$ is an atomic formula. The formulas of \GFO are generated by the following grammar:
\[
\varphi := \top \mid R(\overline{x}) \mid x = y \mid \varphi \land \psi \mid \varphi \lor \psi \mid \lnot \varphi \mid \exists \overline{x} (\alpha \land \varphi),
\]
where, in the last clause, $\alpha$ is a guard for $\varphi$. Note again that we do not allow constants and function symbols.

In the guarded-negation fragment (\GNFO) \cite{Barany2015:guarded}, arbitrary existential quantification is allowed, but every negation is required to be guarded. More precisely, the formulas of \GNFO are generated by the following grammar:
\[ \varphi := \top \mid R(\overline{x}) \mid x = y \mid \varphi \land \varphi \mid \varphi \lor \varphi \mid \exists x \varphi \mid \alpha \land \lnot \varphi, \]
where, in the last clause, $\alpha$ is a guard for $\varphi$.

As is customary, the above definitions are phrased in terms of ordinary guards $\alpha$. However, it is easy to see that if we allow for $\exists$-guards, this would not affect the expressive power (or computational complexity) of these logics in any way. This is because, when the variables in the tuple $\overline{x}$ do not occur free in $\varphi$, as is the case when $\exists \overline{x} \beta$ is an $\exists$-guard for $\varphi$, then we can write $\exists\overline{x}\beta\land \varphi$ equivalently as $\exists\overline{x}(\beta\land\varphi)$. In other words, an $\exists$-guard is as good as an ordinary guard. We call an \FO-formula \emph{self-guarded} if it is of the form $\alpha\land\varphi$ where $\alpha$ is an $\exists$-guard for $\varphi$. Note, in particular, that every sentence $\varphi$ is trivially self-guarded, since it can be equivalently written as $\exists x (x=x) \land \varphi$.

In this section, we will require the notions of \emph{conjunctive queries} (CQs) and \emph{unions of conjunctive queries} (UCQs). A CQ is an \FO-formula of the form
\[ \varphi(x_1, \dots, x_n) := \exists y_1 \dots \exists y_m (\bigwedge_{i \in I} \alpha_i), \]
where each $\alpha_i$ is an atomic formula, possibly an equality, whose free variables are among $\{x_1, \dots, x_n, y_1, \dots, y_m\}$. The collection of all CQs is expressively equivalent to the fragment \FOxc of first-order logic, which is generated by the following grammar:
\begin{equation}
\label{eq:FOxc-grammar}
\varphi := R(x_1, \dots, x_k) \mid x = y \mid \varphi \land \varphi \mid \exists x \varphi.
\end{equation}

A UCQ is a finite disjunction of CQs. Importantly, \GNFO can be alternatively characterized as the smallest logic which can express every UCQ and is closed under guarded negation~\cite{Barany2015:guarded}. This is made explicit in the following expressively equivalent grammar for \GNFO:
\begin{equation}
\label{eq:GNFO-alt-grammar}
\varphi := \top \mid R(\overline{x}) \mid x = y \mid \alpha \land \lnot \varphi \mid q[R_1 / \varphi_1, \dots, R_n / \varphi_n],
\end{equation}
where $q$ is a UCQ with relation symbols $R_1, \dots, R_n$ and $\varphi_1, \dots, \varphi_n$ are self-guarded formulas with the appropriate number of free variables and generated by the same recursive grammar. We refer to this as the UCQ syntax for \GNFO.

\subsection{Natural Extensions of \texorpdfstring{\GFO}{the Guarded Fragment}}

Unlike \FOtwo, guarded fragments are peculiar in that they are not closed under substitution. For example, $\exists xy(R(x,y)\land \neg S(x,y))$ belongs to \GFO, but if we substitute $x=x \land y=y$ for $R(x,y)$, we obtain $\exists xy(x=x\land y=y\land\neg S(x,y))$, which does not belong to \GFO (and is not even expressible in \GNFO). \GFO and \GNFO are, however, closed under \emph{self-guarded substitution}: we can uniformly substitute self-guarded formulas for atomic relations. We generalize the notion of a self-guarded formula to abstract logics $L$ as follows: a formula $\varphi(\overline{x}) \in L(\sigma)$ with $\free(\varphi) = \{x_1,\hdots,x_k\}$ is \emph{self-guarded} if there is a $n$-ary relation symbol $G \in \sigma$, where $n \geq k$, and a tuple of variables $\overline{y}$ containing exactly the variables $\free(\varphi) \cup \{z_1,\hdots,z_m\}$, such that for all $\sigma$-structures $M$ and assignments $g$,
\[ M, g \models \varphi \quad\text{implies}\quad M, g \models \exists z_1 \hdots \exists z_m G(\overline{y}). \]
Intuitively, we can think of a self-guarded $L$-formula as a conjunction of the form $\alpha \land \psi$, where $\alpha$ is an $\exists$-guard for $\psi$. We can then capture the notion of self-guarded substitution for abstract logics by the following definition.

\begin{defi}
\label{def:sg-sub}
We say that an abstract logic $L$ \emph{expresses self-guarded substitutions} if, for each formula $\varphi \in L$ with $R_1,\hdots,R_k \in \sig(\varphi)$, and all self-guarded formulas $\psi_1,\hdots,\psi_k \in L$, we have that $L$ can express the substitution of $\psi_1,\hdots,\psi_k$ for $R_1,\hdots,R_k$ in $\varphi$ (see Definition~\ref{def:abstract-sub}).
\end{defi}

It was shown in~\cite{Barany2015:guarded} that every self-guarded \GFO-formula is expressible in \GNFO. In particular, this applies to all \GFO-sentences and \GFO-formulas with at most one free variable (since all such formulas can be equivalently written as $x=x \land \varphi$). It is therefore common to treat \GNFO as an extension of \GFO. To make this precise, we say that $L'$ is a \emph{self-guarded extension} of $L$ if $L'$ can express all \emph{self-guarded} formulas of $L$ (notation: $L \preceq_{sg} L'$). In Figure~\ref{fig:fragments}, the line marked (*) indicates that \GNFO extends \GFO in this weaker sense. Furthermore, it is worth noting that \GNFO is also not closed under implication, while \GFO is. If it were, then \GNFO would be able to express full negation (using formulas of the form $\varphi \to \bot$). However, \GFO and \GNFO both have disjunction and conjunction in common. We formalize all of these considerations into the following notion.

\begin{defi}
A \emph{guarded logic} is a logic $L$ such that
\begin{enumerate}
\item $\GFO\preceq_{sg} L$, 
\item $L$ expresses self-guarded substitutions, and
\item $L$ expresses conjunction and disjunction.
\end{enumerate}
\end{defi}
\noindent 
Clearly, \GFO and \GNFO are both guarded logics. The \emph{loosely-guarded fragment}, originally introduced in \cite{van1997dynamic}, is a guarded logic as long as it is defined so as to allow for equality guards, as in \cite{Gradel1999decision}. Additional guarded logics include the \emph{clique-guarded fragment} \cite{Gradel1999decision} and the expressively equivalent \emph{packed fragment} \cite{marx2001tolerance}. Furthermore, observe that the \emph{smallest} guarded logic consists of all positive Boolean combinations of self-guarded formulas of \GFO.

\subsection{Finding the Minimal Extension of \texorpdfstring{\GFO}{the Guarded Fragment} with CIP}

Our main result in this section is the following.

\begin{restatable}{thm}{thmmain}
\label{thm:main}
Let $L$ be a guarded logic with CIP. Then $\GNFO \preceq L$.
\end{restatable}
\noindent 
In other words, loosely speaking, \GNFO is the smallest extension of \GFO with CIP. It is based on similar ideas as the proof of Theorem~\ref{thm:FOtwo-main}, but the argument is more intricate. The main thrust of the argument will be to show that our abstract logic $L$ can express all positive existential formulas (Equation \ref{eq:FOxc-grammar}), from which it will follow easily that $L$ is able to express all formulas in the UCQ syntax for \GNFO (Equation \ref{eq:GNFO-alt-grammar}). Toward this end, the main technical result is the following proposition.

\begin{restatable}{prop}{propCQs}
\label{prop:CQs}
Let $L$ be a logic with CIP that can express atomic formulas, guarded quantification, conjunction, and unary implication. Then $\FOxc \preceq L$.
\end{restatable}
\noindent 
Here, we say that a logic $L$ can \emph{express guarded quantification} if, whenever $\varphi\in L$ and $\alpha$ is a guard for $\varphi$, $L$ can express $\exists \overline{x} (\alpha \land \varphi)$; we say that $L$ can \emph{express unary implications} if, whenever $\varphi\in L$ and $\alpha$ is an atomic formula with only one free variable, $L$ can express~$\alpha \to \varphi$. 

The following definition is used in the proof of Proposition~\ref{prop:CQs}.

\begin{defi}
\label{def:bind}
Let $\varphi$ be a formula in \FOxc, let $\overline{y} = y_1, \dots, y_n$ be a tuple of distinct variables, and let $\overline{P}=P_1, \dots, P_n$ be a tuple of unary predicates  of the same length. Then $\bind_{\overline{y}\mapsto \overline{P}}(\varphi)$ is the \FOxc formula given by the following recursive definition:
\begin{enumerate}
\item $\bind_{\overline{y}\mapsto \overline{P}}(\alpha) = \exists \overline{y} \left( \alpha \land \bigwedge_{1 \leq i \leq n : y_i \in \free(\alpha)} P_i(y_i) \right)$, where $\alpha$ is an atomic formula (possibly an equality),
\item $\bind_{\overline{y}\mapsto \overline{P}}(\phi \land \psi) = \bind_{\overline{y}\mapsto \overline{P}}(\phi) \land \bind_{\overline{y}\mapsto \overline{P}}(\psi)$, and
\item $\bind_{\overline{y}\mapsto \overline{P}}(\exists z \psi) = \exists z (\bind_{\overline{y}\mapsto \overline{P}}(\psi))$.
\end{enumerate}
\end{defi}
\noindent 
The $\bind_{\overline{y}\to \overline{P}}$ operation applied to a formula $\varphi \in \FOxc$ wraps each atomic subformula of $\varphi$ with quantifiers for the variables in $\overline{y}$, and adds additional unary predicates for these variables. Our use of the word ``BIND'' is justified by the observation that the free variables of $\bind_{\overline{y} \mapsto \overline{P}}(\varphi)$, for $\overline{y}=y_1,\dots,y_n$, are exactly~$\free(\varphi)\setminus\{y_1, \dots, y_n\}$. The utility of this definition is due to the following fact: for any $\varphi \in \FOxc$, whenever $M, g \models \bind_{\overline{y}\to \overline{P}}(\varphi)$, and the interpretation in $M$ of each unary predicate $P_i$ in $\overline{P}$ is a singleton, then $M, g' \models \varphi$, where $g'$ is the extension of $g$ which maps each $y_i$ to the unique element satisfying $P_i$ (see Propositions \ref{impLemma}, \ref{eqLemma}). The following proposition is a simple consequence of the definition of $\bind$. Note that we write $\varphi \equiv \psi$ to indicate that $\varphi$ and $\psi$ are logically equivalent.

\begin{prop}
\label{BINDcomp}
For all $\FOxc$-formulas $\varphi$ and for all $\overline{x}, \overline{y}$ and $\overline{P},\overline{Q}$, if $\overline{x}$ and $\overline{y}$ are disjoint tuples of distinct variables, and $\overline{P}$ and $\overline{Q}$ are tuples of unary predicates of the same length as $\overline{x}$ and $\overline{y}$, respectively, then $\bind_{\overline{x}\overline{y} \mapsto \overline{P}\overline{Q}}(\varphi) \equiv \bind_{\overline{x} \mapsto \overline{P}}(\bind_{\overline{y} \mapsto \overline{Q}}(\varphi))$.
\end{prop}
\noindent 
A formula $\varphi$ is \emph{clean} if no free variable of $\varphi$ also occurs bound in $\varphi$, and $\varphi$ does not contain two quantifiers for the same variable. Every \FO-formula is equivalent to a clean \FO-formula, and all subformulas of a clean formula are also clean. We now prove two technical propositions.

\begin{prop}
\label{impLemma}
For every clean \FOxc-formula $\varphi$, for every tuple of distinct variables $\overline{y}=y_1, \dots, y_n$ (with each $y_i\in \free(\varphi)$), and for every tuple of unary predicates $\overline{P}=P_1,\ldots,P_n$, we have that
\[
\big(\bigwedge_{i=1,\hdots,n} P_i(y_i)\big)\models \varphi  \to \bind_{\overline{y} \mapsto \overline{P}}(\varphi).
\]
\end{prop}
\begin{proof}
The proof is by induction on the complexity of $\varphi$. For the base case, suppose that $\varphi$ is an atomic formula. Then if $M,g \models \bigwedge_{i=1\ldots n} P_i(y_i)$ and $M,g \models \varphi$, then it is immediate from the definition of $\bind$ that $M,g \models \bind_{\overline{y} \mapsto \overline{P}}(\varphi)$. For the inductive step, suppose that, for every model $M$ and variable assignment $g$, if $M,g\models \bigwedge_{i=1\ldots n} P_i(y_i)$ and $M,g\models\varphi$ then $M,g\models \bind_{\overline{y} \mapsto \overline{P}}(\varphi)$. The case of conjunction is trivial by the inductive hypothesis and the fact that $\bind$ commutes with conjunction. Finally, suppose that $\varphi := \exists z \psi$. By the inductive hypothesis, we have that
\[ \big(\bigwedge_{i=1,\hdots,n} P_i(y_i)\big)\models \psi  \to \bind_{\overline{y} \mapsto \overline{P}}(\psi). \]
Suppose that $M,g \models \bigwedge_{i=1,\hdots,n} P_i(y_i)$ and $M,g \models \exists z \psi$. Then there exists some $a \in M$ such that $M,g[z/a] \models \psi$, and so by the inductive assumption, $M,g[z/a] \models \bind_{\overline{y} \mapsto \overline{P}}(\psi)$. Then because $\varphi$ is clean, we have that $z \not \in \{y_1,\hdots,y_n\}$, and so $M,g \models \exists z\bind_{\overline{y} \mapsto \overline{P}}(\psi)$.
\end{proof}

\begin{prop}
\label{eqLemma}
For every clean \FOxc-formula $\varphi(x,\overline{y})$ with $\overline{y}=y_1, \ldots, y_n$ distinct from $x$, and for every $n$-tuple of unary predicates
$\overline{P}=P_1, \ldots, P_n$ not occurring in $\varphi$, we have that
\[ \exists x\varphi(x,\overline{y}) \equiv \forall \overline{P}\Big(\big(\bigwedge_{i=1\ldots n} P_i(y_i)\big) \to \exists x \bind_{\overline{y} \mapsto \overline{P}}(\varphi(x,\overline{y}))\Big). \]
\end{prop}
\begin{proof} For the left-to-right direction, suppose that $M,g \models \exists x \varphi(x,\overline{y})$ and $M'$ is an expansion of $M$ such that $M' \models \bigwedge_{i=1 \hdots n} P_i(y_i)$. Then $M',g[x/b] \models \varphi(x,\overline{y}) \land \bigwedge_{i=1\ldots n} P_i(y_i)$ for some $b\in M$. Then, by Proposition~\ref{impLemma}, we have
\[ M', g[x/b] \models \bind_{\overline{y} \mapsto \overline{P}}(\varphi(x,\overline{y})), \]
and hence
$M',g\models \exists x\bind_{\overline{y}\mapsto\overline{P}}(\varphi(x,\overline{y}))$.

For the reverse direction, suppose that
\[ M,g \models \forall \overline{P}(\bigwedge_i P_i(y_i) \to \exists x \bind_{\overline{y} \mapsto \overline{P}}(\varphi(x,\overline{y}))). \]
Let $M'$ be the expansion
of the structure $M$ in which each unary predicate symbol $P_i$ is interpreted as $\{g(y_i)\}$. Then, we have that $M',g\models \exists x \bind_{\overline{y} \mapsto \overline{P}}(\varphi(x,\overline{y}))$, and hence $M',g[x/b] \models \bind_{\overline{y} \mapsto \overline{P}}(\varphi(x,\overline{y}))$ for some $b \in M$. 
To complete the proof, it suffices to show that
$M',g[x/b]\models \varphi(x,\overline{y})$ (since 
this implies that also $M,g[x/b]\models\varphi(x,\overline{y})$).

For every subformula containing a bound occurrence of a variable $y_i \in \overline{y}$, we have that every witness for that variable $y_i$ must also be in $P_i$ (by construction of $\bind_{\overline{y} \mapsto \overline{P}}(\varphi(x,\overline{y}))$ and the assumption that $\varphi(x,\overline{y})$ is clean). Since each $P_i$ is a singleton, this implies that each witness for $y_i$ in each subformula is $g(y_i)$. It follows that $M', g[x/b] \models \alpha$ for each atomic formula $\alpha$ occurring in $\varphi(x,\overline{y})$. By a simple subformula induction, we then obtain that $M', g[x/b] \models \varphi(x,\overline{y})$, completing the proof.
\end{proof}

\noindent The following lemma enables the proof of Proposition~\ref{prop:CQs}.

\begin{lem}
\label{BINDexpLemma}
Let $L$ be an \FO-fragment which can express atomic formulas and is closed under guarded quantification, conjunction, and unary implication. If $L$ can express a formula $\varphi \in \FOxc$ and all of its subformulas, then for all tuples $\overline{y}$ of variables, we have that $L$ can express $\bind_{\overline{y} \mapsto \overline{P}}(\varphi)$.
\end{lem}

\begin{proof}
We proceed by strong induction on the complexity of clean \FOxc-formulas $\varphi$.

\textbf{Base Case.} Suppose $\varphi$ is an atomic formula. Fix an arbitrary tuple $\overline{y}=y_1 \ldots, y_n$.~Then
\[
\bind_{\overline{y} \mapsto \overline{P}}(\varphi) = \exists \overline{y} \left( \alpha \land \bigwedge_{1 \leq i \leq n : y_i \in \free(\alpha)} P_i(y_i) \right),
\]
which $L$ can express by closure under conjunction and guarded quantification.

\textbf{Inductive Step.} Suppose inductively that, for all formulas $\psi$ of lesser complexity than $\varphi$, and all tuples $\overline{z}$ of variables, we have that $L$ can express $\bind_{\overline{z} \mapsto \overline{P}}(\psi)$. Fix an arbitrary tuple $\overline{y}$ of variables.

Suppose that $\varphi = \psi_1 \land \psi_2$. Since $L$ can express $\varphi$ and all of its subformulas, it can also express $\psi_1$, $\psi_2$, and all of their subformulas. Then by the inductive hypothesis, $L$ can express $\bind_{\overline{y} \mapsto \overline{P}}(\psi_1)$ and $\bind_{\overline{y} \mapsto \overline{P}}(\psi_2)$. Then by closure under conjunctions, $L$ can express $\bind_{\overline{y} \mapsto \overline{P}}(\psi_1) \land \bind_{\overline{y} \mapsto \overline{P}}(\psi_2)$, which is the same as $\bind_{\overline{y} \mapsto \overline{P}}(\varphi)$ (see Definition \ref{def:bind}).

\hspace*{1pt} \\
Now suppose that $\varphi(\overline{x},\overline{y}) = \exists z \psi(\overline{x},\overline{y},z)$, where the (possibly empty) tuple $\overline{x}$ consists of all free variables of $\varphi$ not in the tuple $\overline{y}$. We need to show that $L$ can express $\bind_{\overline{y} \mapsto \overline{P}}(\varphi(\overline{x},\overline{y}))$, which is the same as $\exists z (\bind_{\overline{y} \mapsto \overline{P}}(\psi(\overline{x},\overline{y},z)))$ (see Definition \ref{def:bind}). Since $L$ can express $\varphi$ and all of its subformulas, it can also express $\psi$ and all of its subformulas. Then, by the inductive hypothesis, 
$L$ can express $\bind_{\overline{y} \mapsto \overline{P}}(\psi)$, whose free variables are those in the tuple $\overline{x}z$, as well as $\bind_{\overline{x}\overline{y} \mapsto \overline{Q}\overline{P}}(\psi)$, whose only free variable is $z$, where $\overline{Q}$ and $\overline{P}$ are distinct tuples of fresh unary predicates with $\overline{Q}$ the length of $\overline{x}$ and $\overline{P}$ the length of $\overline{y}$. Since $L$ is closed under conjunction and guarded quantification, it follows that $L$ can express
\[\gamma(\overline{x}) :=\exists z (G(\overline{x},z) \land \bind_{\overline{y} \mapsto \overline{P}}(\psi))
\quad\text{and}\quad
\exists z (z=z \land \bind_{\overline{x}\overline{y} \mapsto \overline{Q}\overline{P}}(\psi)),\]
where $G$ is a fresh relation symbol not occurring in $\psi$. Then by closure under unary implications, we have that $L$ can also express
\[ \chi(\overline{x}) := \big(\bigwedge_{i} Q_i(x_i)\big) \to \exists z (z=z \land \bind_{\overline{x}\overline{y} \mapsto \overline{Q}\overline{P}}(\psi)). \]

\medskip\par\noindent\textbf{Claim: }
$\gamma(\overline{x})\models \chi(\overline{x})$

\medskip\par\noindent\emph{Proof of claim:}
By Proposition \ref{BINDcomp}, 
\begin{equation}
\bind_{\overline{x}\overline{y} \mapsto \overline{Q}\overline{P}}(\psi) \equiv \bind_{\overline{x} \mapsto \overline{Q}}(\bind_{\overline{y} \mapsto \overline{P}}(\psi)).
\label{eq:bind-bind}
\end{equation} 
Then by applying Proposition \ref{impLemma} and inverting the hypotheses, we have
\[ \bind_{\overline{y} \mapsto \overline{P}}(\psi) \models \big(\bigwedge_{i} Q_i(x_i)\big) \to  \bind_{\overline{x}\overline{y} \mapsto \overline{Q}\overline{P}}(\psi). \]
From this, it  follows (because $z$ is distinct from $x_i$ variables) that
\[ \exists z(\bind_{\overline{y} \mapsto \overline{P}}(\psi)) \models \big(\bigwedge_{i} Q_i(x_i)\big) \to \exists z\bind_{\overline{x}\overline{y} \mapsto \overline{Q}\overline{P}}(\psi), \]
and therefore $\gamma(\overline{x})\models \chi(\overline{x})$. This concludes the proof of the claim.

\medskip
Since $L$ can express both $\gamma(\overline{x})$ and $\chi(\overline{x})$, we have by the Craig interpolation property that $L$ can express some Craig interpolant $\vartheta(\overline{x})$. Since $G$ and the $Q_i$ predicates do not occur in $\varphi$, they do not occur in $\vartheta(\overline{x})$, and therefore, the following second-order implication is valid:
\[ \exists G \gamma(\overline{x}) \models \vartheta(\overline{x}) \models \forall \overline{Q} \chi(\overline{x}). \]

It is easy to see that $\exists G \gamma(\overline{x}) \equiv \exists z \bind_{\overline{y} \mapsto \overline{P}}(\psi)$. Similarly, it follows from Proposition \ref{eqLemma} and Equation (\ref{eq:bind-bind}) that $\forall \overline{Q} \chi(\overline{x}) \equiv \exists z \bind_{\overline{y} \mapsto \overline{P}}(\psi)$.
Therefore, $\vartheta(\overline{x}) \equiv \exists z \bind_{\overline{y} \mapsto \overline{P}}(\psi)$. In particular,
$\exists z \bind_{\overline{y} \mapsto \overline{P}}(\psi)$ is expressible in $L$.
\end{proof}

\noindent We are now ready to prove Proposition~\ref{prop:CQs}, restated below.

\propCQs*

\begin{proof} By strong induction on the complexity of \FOxc-formulas. The base case is immediate, since $L$ can express all atomic formulas. For the inductive step, if $\varphi := \psi_1 \land \psi_2$, then by the inductive hypothesis, $L$ can express $\psi_1$ and $\psi_2$, and so by closure under conjunction, $L$ can express $\varphi$. Now suppose $\varphi(\overline{y}) := \exists x(\psi(x,\overline{y}))$.
By the inductive hypothesis, together with closure
under guarded quantification, $L$ can express
\[\gamma(\overline{y}) := \exists x (G(x, \overline{y}) \land \psi).\]
Furthermore, by Lemma \ref{BINDexpLemma}, $L$ can express $\bind_{\overline{y} \mapsto \overline{P}}(\psi)$, and therefore, by closure under guarded quantification and unary implications, $L$ can express \[\chi(\overline{y}) := \big(\bigwedge_{i} P_i(y_i)\big) \to \exists x(x=x\land \bind_{\overline{y} \mapsto \overline{P}}(\psi)).\] 

\medskip\par\noindent\textbf{Claim: } 
$\gamma(\overline{y}) \models \chi(\overline{y})$.

\medskip\par\noindent\emph{Proof of claim:}
It is clear that $\gamma(\overline{y})\models \exists x \psi$.
Furthermore, by Proposition \ref{impLemma},
$\psi\models\big(\bigwedge_{i} P_i(y_i)\big) \mapsto \bind_{\overline{y} \mapsto \overline{P}}(\psi)$, from which it follows that
$\exists x \psi \models\chi(\overline{y})$
(since the variable $x$ is distinct from $y_1, \ldots, y_n$).
Therefore, $\gamma(\overline{y}) \models \chi(\overline{y})$.
\medskip

Let $\vartheta(\overline{y})$ be an interpolant for $\gamma(\overline{y}) \models \chi(\overline{y})$ in $L$.
Since $G$ and the predicates in $\overline{P}$ do not occur in $\psi$, the following second-order entailments are valid:
\[
\exists G \exists x (G(x, \overline{y}) \land \psi) \models \vartheta(\overline{y}) \models \forall \overline{P} ((\bigwedge_{i} P_i(y_i)) \to \exists x \bind_{\overline{y} \mapsto \overline{P}}(\psi)).
\]
It is easy to see that $\exists G \exists x (G(x, \overline{y}) \land \psi) \equiv \exists x \psi$. Furthermore, by Proposition \ref{eqLemma},
\[
\psi \equiv \forall \overline{P} ((\bigwedge_{i} P_i(y_i)) \to \bind_{\overline{y} \mapsto \overline{P}}(\psi)).
\]
from which it follows (since $x$ is distinct from $y_1, \ldots, y_n$) that
\[
\exists x \psi \equiv \forall \overline{P} ((\bigwedge_{i} P_i(y_i)) \to \exists x \bind_{\overline{y} \mapsto \overline{P}}(\psi)).
\]
Therefore, $\vartheta(\overline{y}) \equiv \varphi(\overline{y})$, and so we are done.
\end{proof}

\noindent Our main result follows easily from Proposition \ref{prop:CQs}, the closure properties of guarded logics, and the UCQ characterization of \GNFO (Equation \ref{eq:GNFO-alt-grammar}).

\thmmain*

\begin{proof}
$L$ can express self-guarded \GFO-formulas, so it can express formulas of the form $\exists \overline{x} \beta$, where $\beta$ is an atomic formula. Then since $L$ can express self-guarded substitution, $L$ can express guarded quantification. Furthermore, $L$ can express all self-guarded formulas of the form $\alpha \land \lnot \beta$, where $\alpha$ and $\beta$ are atomic formulas such that $\free(\alpha) = \free(\beta)$. Furthermore, for every formula $\varphi$ expressible in $L$ with $\free(\varphi) \subseteq \free(\alpha)$, $\alpha \land \varphi$ is a self-guarded formula. Thus by expressibility of self-guarded substitution, $L$ can also express $\alpha \land \lnot (\alpha \land \varphi)$, which is equivalent to $\alpha \land \lnot \varphi$; hence $L$ can express guarded negation. If $L$ can express $\varphi$, then by expressibility of guarded negation and disjunction, it can also express the formula $(x=x \land \lnot P(x)) \lor \varphi$, which is equivalent to $P(x) \to \varphi$. Hence $L$ can express unary implications. Therefore, by Proposition \ref{prop:CQs}, $L$ can express all formulas in \FOxc. Then by expressibility of disjunction, $L$ can express all unions of conjunctive queries. The result then follows immediately from the UCQ-syntax for \GNFO (Equation \ref{eq:GNFO-alt-grammar}), by closure under self-guarded substitution.
\end{proof}

\section{Repairing Interpolation for \texorpdfstring{\FF}{the Forward Fragment}}
\label{sec:rep-ff}
The fluted fragment (\FL) \cite{quine1969} is an ordered logic, in which all occurrences of variables in atomic formulas and quantifiers must follow a fixed order. In the context of ordered logics, we assume a fixed infinite sequence of variables $X = \langle x_i \rangle_{i \in \mathbb{Z}^+}$. A \emph{suffix $n$-atom} is an atomic formula of the form $R(x_j,\dots,x_n)$, where $x_j, \dots, x_n$ is a finite contiguous subsequence of $X$. \FL is defined by the following recursion.
\looseness=-1

\begin{defi}
For each $n \in \mathbb{N}$, define collections of formulas $\FL^n$ as follows:
\begin{enumerate}
\item $\FL^n$ contains all suffix $n$-atoms,
\item $\FL^n$ is closed under Boolean combinations, and
\item If $\varphi$ is in $\FL^{n+1}$, then $\exists x_{n+1} \varphi$ and $\forall x_{n+1} \varphi$ are in $\FL^n$.
\end{enumerate}
We set $\FL = \bigcup_{n \in \mathbb{N}} \FL^n$.
\end{defi}

The forward fragment (\FF), introduced in \cite{bednarczyk2021}, is a syntactic generalization of \FL. We say that $R(x_j,\dots,x_k)$ is an \emph{infix $n$-atom} if $x_j, \dots, x_n$ is a finite contiguous subsequence of $X$ and $k \leq n$. \FF is defined by the following recursion.

\begin{defi}
For each $n \in \mathbb{N}$, define collections of formulas $\FF^n$ as follows:
\begin{enumerate}
\item $\FF^n$ contains all infix $n$-atoms,
\item $\FF^n$ is closed under Boolean combinations, and
\item If $\varphi$ is in $\FF^{n+1}$, then $\exists x_{n+1} \varphi$ and $\forall x_{n+1} \varphi$ are in $\FF^n$.
\end{enumerate}
We set $\FF = \bigcup_{n \in \mathbb{N}} \FF^n$.
\end{defi}

In contrast to the other logics we have seen, \FL and \FF do not allow the primitive equality symbol. It can be seen by a simple formula induction that every formula in $\FF^k$ can be expressed by a formula in $\FF^n$ for every $n > k$; it follows easily that \FF can express arbitrary Boolean combinations of its formulas. However, \FL cannot: $P(x_1)$ and $P(x_2)$ are in $\FL$, but $P(x_1) \land P(x_2)$ is not expressible in \FL. Although \FF contains formulas which are not in \FL, it is known that \FF and \FL are expressively equivalent at the level of sentences~\cite{bednarczyk2022}. Furthermore, the satisfiability problems for \FL and \FF are decidable~\cite{purdy1996-2,bednarczyk2022}.

\subsection{Natural Extensions of \texorpdfstring{\FF}{the Fluted Fragment}}
Given a formula $\varphi$, we write $\gfv(\varphi)$ to denote the greatest $n \in \mathbb{Z}^+$ such that $x_n$ occurs free in $\varphi$; if $\varphi$ is a sentence, then we set $\gfv(\varphi) = 0$. We define \emph{forward logics} to capture the notion of a natural extension of \FF.

\begin{defi}
A \emph{forward logic} is an abstract logic $L$ such that
\begin{enumerate}
\item $L$ can express all infix $n$-atoms for every $n \in \mathbb{Z}^+$,
\item $L$ can express all Boolean combinations of its formulas, and
\item $L$ can express $\exists x_n \varphi$ and $\forall x_n \varphi$ whenever $L$ can express $\varphi$ and $n = \gfv(\varphi)$.
\end{enumerate}
\end{defi}
\noindent 
We refer to the last property of a forward logic as \emph{expressibility of ordered quantification}. Clearly \FF is a forward logic, and every forward logic extends \FF. Aside from \FF itself, the recently-introduced \emph{adjacent fragment} (\AF) syntactically generalizes both \FL and \FF~\cite{bednarczyk2023adjacent}, and is itself a forward logic in our sense. 

\subsection{Finding the Minimal Extension of \texorpdfstring{\FF}{the Fluted Fragment} with CIP}
Unlike the other fragments we have seen, one peculiar property of \FF is that the logic is not closed under variable substitutions. This can be seen simply by considering relational atoms: for a 3-ary relational symbol $R$, the formula $R(x_1,x_2,x_3)$ is in \FF, but the formula $R(x_3,x_1,x_2)$ is not. Before proving our main theorem, we prove the following lemma asserting that whenever a formula is expressible in a forward logic $L$ satisfying CIP, the result of making arbitrary substitutions for the free variables of the formula is also expressible in $L$.

\begin{lem}
\label{lemma:ff-shuffle}
Let $L$ be a forward logic satisfying CIP, and let $\varphi(x_{i_1}, \dots, x_{i_k})$ be a formula of first-order logic expressible in $L$, where $x_{i_1}, \dots, x_{i_k}$ is not necessarily a contiguous subsequence of variables. Then for every map
\[
\pi: \{i_1,\dots,i_k\} \to \mathbb{Z}^+,
\]
we have that $L$ can also express $\varphi(x_{\pi(i_1)}, \dots, x_{\pi(i_k)})$. In other words, $L$ is closed under renamings of free variables.
\end{lem}
\begin{proof}
For brevity, let $\overline{x} = x_{i_1}, \dots, x_{i_k}$, and let $\pi(\overline{x}) = x_{\pi(i_1)}, \dots, x_{\pi(i_k)}$. Without loss of generality, assume that $i_1 \leq \dots \leq i_k$ (we can do this since the notation $\varphi(x_{i_1}, \dots, x_{i_k})$ only indicates that the variables occur free, but says nothing about where or in what order they occur in the formula). Since $L$ can express $\varphi(\overline{x})$, it can evidently express the following formulas, by the definition of a forward logic:
\begin{align*}
\gamma(\overline{x}) &:= \bigwedge_{m \leq k} G_m(x_{i_m}) \land \forall x_{i_1} \dots \forall x_{i_k} \left( \bigwedge_{m \leq k} G_m(x_{\pi(i_m)}) \to \varphi(\overline{x}) \right) \\
\chi(\overline{x}) &:= \bigwedge_{m \leq k} P_m(x_{i_m}) \to \exists x_{i_1} \dots \exists x_{i_k} \left( \varphi(\overline{x}) \land \bigwedge_{m \leq k} P_m(x_{\pi(i_m)}) \right)
\end{align*}
Clearly $\gamma \models \chi$, and so there exists an interpolant $\vartheta$. Hence
\[
\exists G_1 \hdots G_k \gamma \models \vartheta \models \forall P_1 \hdots P_k \chi
\]
is a valid second-order entailment. Furthermore, it is easy to see that
\[
\exists G_1 \hdots G_k \gamma \equiv \forall P_1 \hdots P_k \chi \equiv \varphi.
\]
Therefore $\vartheta$ is equivalent to $\varphi(x_{\pi(i_1)}, \dots, x_{\pi(i_k)})$ and is expressible in $L$.
\end{proof}

\noindent We now prove our main theorem, which follows easily from Lemma \ref{lemma:ff-shuffle}.

\begin{thm}
Let $L$ be a forward logic satisfying CIP. Then $\FO \preceq L$.
\end{thm}
\begin{proof}
We proceed by induction on the complexity of \FO-formulas $\varphi$. For the base case, clearly $L$ can express all atomic \FO-formulas by applying Lemma \ref{lemma:ff-shuffle} to an appropriate infix atom. For the inductive step, the Boolean cases are immediate since $L$ can express all Boolean combinations. Hence the only interesting case is when $\varphi = \exists x_k \psi$ for some formula $\psi$. By the inductive hypothesis, $L$ can express $\psi$. Applying Lemma \ref{lemma:ff-shuffle}, $L$ can also express $\varphi'$, the result of substituting $x_{n+1}$ for all free occurrences of $x_k$, where $n = \gfv(\varphi)$, and leaving all other free variables the same. Then by expressibility of ordered quantification, $L$ can express $\exists x_{n+1} \varphi'$, which is equivalent to $\varphi$.
\end{proof}

\section{Undecidability of Extensions of \texorpdfstring{\FOtwo}{the Two-Variable Fragment} and \texorpdfstring{\FL}{the Fluted Fragment} with CIP}
\label{sec:weak-ext}
In Section \ref{sec:rep-fo2}, we showed that every strong extension of \FOtwo with CIP can express all sentences of \FO, and in Section \ref{sec:rep-ff}, we showed that every forward logic with CIP can express all formulas of \FO. These results suggest the undecidability of the satisfiability problems for such logics. In this section, we formalize this idea, showing that extensions of \FOtwo and \FL with CIP and satisfying very limited expressive assumptions are undecidable. These results rely primarily on known results on the undecidability of \FOtwo and \FL with additional transitive relations.

\begin{prop}
\label{prop:almost-trans}
Every abstract logic $L$ with CIP extending \FOtwo or \FL can express the following formulas:
\begin{align*}
\psi_0(x_1) &:= \forall x_2 \forall x_3 (R(x_1,x_2) \land R(x_2,x_3) \to R(x_1,x_3)), ~\text{and} \\
\psi_1 &:= \lnot \forall x_1 \forall x_2 \forall x_3 (R(x_1,x_2) \land R(x_2,x_3) \to R(x_1,x_3)).
\end{align*}
\end{prop}
\begin{proof}
We will obtain $\psi_0$ and $\psi_1$ as the unique interpolants of valid entailments between formulas of \FOtwo and \FL. Toward this end, consider the following formulas of \FL:
\begin{align*}
\gamma_0(x_1) &:= \forall x_2 (R(x_1,x_2) \to \forall x_3 (R(x_2,x_3) \to G(x_3)) \land \forall x_2 (G(x_2) \to R(x_1,x_2))), \\
\gamma_1(x_1) &:= \exists x_2 (R(x_1,x_2) \land \exists x_3 (R(x_2,x_3) \land P(x_3))) \to \exists x_2 (R(x_1,x_2) \land P(x_2)).
\end{align*}
We also define the following sentences of \FL:
\begin{align*}
\delta_0 := &\forall x_1 ( G_1(x_1) \to \forall x_2 (G_2(x_2) \to \lnot R(x_1,x_2))) \\
&\hspace*{25pt} \land \exists x_1 (G_1(x_1) \land \exists x_2 (R(x_1,x_2) \land \exists x_3 (R(x_2,x_3) \land G_2(x_3)))), \\
\delta_1 := &\exists x_1 ( \forall x_2 (P(x_2) \leftrightarrow \lnot R(x_1,x_2)) \to \\
&\hspace*{25pt} \exists x_2 ( R(x_1,x_2) \land \exists x_3 (R(x_2,x_3) \land P(x_3)))).
\end{align*}
With simple variable substitutions, we can turn $\gamma_0$ and $\gamma_1$ into formulas of \FOtwo:
\begin{align*}
\gamma'_0(x) &:= \forall y (R(x,y) \to \forall x (R(y,x) \to G(x)) \land \forall y (G(y) \to R(x,y))), \\
\gamma'_1(x) &:= \exists y (R(x,y) \land \exists x (R(y,x) \land P(x))) \to \exists y (R(x,y) \land P(y)).
\end{align*}
Similarly, we can turn $\delta_0$ and $\delta_1$ into sentences of \FOtwo:
\begin{align*}
\delta'_0 := &\forall x ( G_1(x) \to \forall y (G_2(y) \to \lnot R(x,y))) \\
&\hspace*{25pt} \land \exists x (G_1(x) \land \exists y (R(x,y) \land \exists x (R(y,x) \land G_2(x)))), \\
\delta'_1 := &\exists x ( \forall y (P(x) \leftrightarrow \lnot R(x,y)) \to \\
&\hspace*{25pt} \exists y ( R(x,y) \land \exists x (R(y,x) \land P(x)))).
\end{align*}
It is a simple exercise in \FO-semantics to see that $\gamma_0 \models \gamma_1$ and $\delta_0 \models \delta_1$. Since we assumed that $L$ has CIP, the following claims complete the argument.

\begin{clm}
$\psi_0$ is the unique interpolant of $\gamma_0$ and $\gamma_1$.
\end{clm}
\begin{proof}
Let $\vartheta(x_1)$ be an $L$-interpolant, so that $\gamma_0 \models \vartheta \models \gamma_1$; then $\exists G \gamma_0 \models \vartheta \models \forall P \gamma_1$ is also a valid second-order entailment. To complete the proof, we need to show that $\exists S \gamma_0 \equiv \psi_0$ and $\forall P \gamma_1 \equiv \psi_0$. To see that $\psi_0 \models \exists S \gamma_0$, observe that if $M,g \models \psi_0$, then $M',g \models \gamma_0$, where $M'$ is the expansion of $M$ in which $S$ is interpreted as the singleton set $\{g(x_0)\}$. If $M,g \models \forall P \gamma_1$, then $M',g \models \gamma_1$, where $M'$ is the expansion of $M$ in which $S$ is interpreted as the singleton set $\{g(x_0)\}$; hence $M',g \models \psi_0$, and so $M,g \models \psi_0$. Thus we have that $\psi_0 \models \exists S \gamma_0 \models \vartheta \models \forall P \gamma_1 \models \psi_0$; in particular, we have that $\vartheta \equiv \psi_0$, completing the~proof.
\end{proof}

\begin{clm}
$\psi_1$ is the unique interpolant of $\delta_0$ and $\delta_1$.
\end{clm}
\begin{proof}
Let $\vartheta(x_1)$ be an $L$-interpolant, so that $\delta_0 \models \vartheta \models \delta_1$; then $\exists G_1 G_2 \delta_0 \models \vartheta \models \forall P \delta_1$ is also a valid second-order entailment. To complete the proof, we need to show that $\exists G_1 \exists G_2 \delta_0 \equiv \psi_1$ and $\forall P \delta_1 \equiv \psi_0$. To see that $\psi_1 \models \exists G_1 \exists G_2 \delta_0$, observe that if that $M,g \models \psi_1$, then $M',g \models \delta_0$, where $M'$ is the expansion of $M$ in which $G_1$ and $G_2$ are interpreted as singleton sets containing the witnesses for the first-order existential quantifiers in $\psi_1$. If $M, g \models \forall P \delta_1$, then $M', g \models \delta_1$, where $M'$ is the expansion of $M$ in which $P$ is interpreted as the set of elements in $a \in \dom(M)$ such that $R^M(g(x_1),a)$ does not hold. It then follows easily that $M \models \psi_1$.
\end{proof}
\noindent This concludes the proof.
\end{proof}

We also need two additional definitions. First, an \emph{effective translation} from a logic $L$ to a logic $L'$ is a computable function which takes formula of $\varphi \in L$ as input and outputs an equivalent formula $\varphi' \in L'$. Second, we say that a logic $L$ has \emph{effective conjunction} if there is a computable function taking formulas $\varphi, \psi \in L$ as input and outputting a formula $\chi \in L$ which is equivalent to $\varphi \land \psi$.

\begin{thm}
\label{thm:ff-weak-und}
Let $L$ be an extension of \FL which satisfies CIP. Suppose further that there is an effective translation from \FL to $L$, and $L$ has effective conjunction. The satisfiability problem for $L$ is undecidable if either
\begin{enumerate}
\item $L$ can express ordered quantification, or
\item $L$ can express negation.
\end{enumerate}
\end{thm}
\begin{proof}
Let $\chi$ be the sentence asserting the transitivity of the relation $R$. Since $L$ has CIP and extends \FL, it can express both $\psi_0(x_1)$ and $\psi_1$ by Proposition \ref{prop:almost-trans}. If $L$ can express ordered quantification, it can express $\forall x_1 \psi_0(x_1)$, which is equivalent to $\chi$. If $L$ can express negation, then it can express $\lnot \psi_1$, which is also equivalent to $\chi$. As an abstract logic, $L$ can express $\chi$ and is closed under predicate renamings, so it can express that any number of binary relations are transitive. Let $\chi_1$, $\chi_2$, and $\chi_3$ be sentences expressing transitivity of binary relation symbols $R_1$, $R_2$, and $R_3$, respectively. Let $tr$ be an effective translation from \FL to $L$. Then a formula $\varphi$ of \FL with three designated transitive relations is satisfiable if and only if $tr(\varphi) \land \chi_1 \land \chi_2 \land \chi_3$ is satisfiable. Since $tr$ is computable and $L$ is effectively closed under conjunction, this reduction is computable. Satisfiability of \FL with three transitive relations is undecidable \cite{pratt-hartmann2022}, and so satisfiability of $L$ is undecidable.
\end{proof}

Note that notions of ``effective ordered quantification'' and ``effective negation,'' analogous to ``effective conjunction,'' are not required in the preceding proof, since the expressibility of those connectives is needed only to guarantee that transitivity of relations can be expressed.

Furthermore, it is also known that satisfiability is undecidable for $FO^2$-formulas with two transitive relations~\cite{Kieronski2005Results}. Using this fact, along with Proposition \ref{prop:almost-trans}, we obtain the following theorem, by a similar proof to that of Theorem \ref{thm:ff-weak-und}. Note again that we do not require a notion of ``effective universal quantification'' in the following proof.
\begin{thm}
\label{thm:fo2-weak-und}
Let $L$ be an extension of \FOtwo which satisfies CIP. Suppose further that there is an effective translation from \FOtwo to $L$, and $L$ has effective conjunction. The satisfiability problem for $L$ is undecidable if either
\begin{enumerate}
\item $L$ can express universal quantification, or
\item $L$ can express negation.
\end{enumerate}
\end{thm}
\noindent 
We remark that all forward logics and strong extensions of \FOtwo with CIP, assuming appropriate effective translations and effective conjunction, meet the requirements of Theorems \ref{thm:ff-weak-und} and \ref{thm:fo2-weak-und}, and hence are undecidable.

\section{Discussion}
\label{sec:conclusion}
In the introduction, we mentioned several results indicating the failure of CIP among many natural proper extensions of \FO. In \cite{vanbenthem2008}, van Benthem points out that there is a similar scarcity among \FO-fragments as well. Our results in Sections \ref{sec:rep-fo2} and \ref{sec:rep-ff} may be interpreted as additional confirmation of this observation. Furthermore, one tends to study proper fragments of \FO for their desirable computational properties, and so our broader undecidability results show that CIP fails for large swaths of \emph{decidable} \FO-fragments. However, there are a few notable fragments for which the determination of a minimal extension satisfying CIP is still open, such as \FL and some of the quantifier prefix fragments.

One limitation of our methodology and results is their dependence on a definition of Craig interpolation which mandates the existence of interpolants between proper \emph{formulas}, while many practical applications only require CIP for \emph{sentences}. Throughout this paper, we have established expressibility of a formula $\vartheta$ in a logic $L$ by induction (and by constructing two formulas $\varphi$ and $\psi$ such that $\varphi \models \psi$ and arguing that every interpolant is equivalent to $\vartheta$). In general, this method is difficult to apply unless free variables are allowed; it is not clear how to apply this type of inductive argument if we were only concerned with the existence of interpolants for sentences of the logic.

There are several well-studied properties strictly weaker than CIP. The $\Delta$-interpolation property (also known as Suslin-Kleene interpolation) holds for a logic $L$ if, whenever $\varphi \models \psi$, and (intuitively speaking) there is only one possible interpolant $\vartheta$ up to logical equivalence for this entailment, then $L$ contains a formula equivalent to $\vartheta$~\cite{barwise2017}. It is not hard to see that, unlike the Craig interpolation property, every logic $L$ has a unique minimal extension, denoted $\Delta(L)$, satisfying the $\Delta$-interpolation property. In fact, in our proofs we only rely on $\Delta$-interpolation; every application of the assumption that some abstract logic $L$ satisfies CIP yields a provably unique interpolant, up to logical equivalence. Therefore, all of our results hold also when CIP is replaced by $\Delta$-interpolation.

Two additional weakenings of CIP are the projective, non-projective, and weak Beth definability properties. The \emph{projective Beth property} states, roughly, that whenever a $\sigma\cup\tau\cup\{R\}$-theory $\Sigma$ ``implicitly defines'' a relation $R$ in terms of the relations in $\sigma$, then $\Sigma$ entails an ``explicit definition'' of $R$ in terms of $\sigma$. The \emph{(non-projective) Beth property} is the special case for $\tau=\emptyset$. The \emph{weak Beth property} is a further weakening, where 
an explicit definition of $R$ in terms of $\sigma$ is required to exist only when 
every $\sigma$-structure has a \emph{unique}
$\sigma\cup\{R\}$-expansion satisfying $\Sigma$. Many practical applications of CIP in database theory and knowledge representation require only the projective Beth property. It is not immediately clear how to extend our methodology to a systematic study of the (projective) Beth property among decidable \FO-fragments. Indeed, \GFO already satisfies the non-projective Beth property \cite{Hoogland02:interpolation}, while \FOtwo satisfies the weak Beth definability property \cite{andreka2021two}. Given their applications, an interesting avenue of future work is to map the landscape of \FO-fragments satisfying these properties.

In the other direction, minimal
extensions of logics with \emph{uniform} 
interpolation (a strengthening of CIP) were studied
in~\cite{DAgostino2006}, although 
with limited results so far
(cf.~\cite[Thm.~14]{DAgostino2006}). Some of the minimal extensions of \PLTL fragments with CIP identified in~\cite{gheerbrant2009craig}, however, do satisfy uniform interpolation.
\looseness=-1

\section*{Acknowledgment}
\noindent We thank Jean Jung, Frank Wolter, and Malvin Gattinger for feedback on an earlier draft, and we thank Ian Pratt-Hartmann and Michael Benedikt for helpful remarks during a related workshop presentation. Balder ten Cate is supported by EU Horizon
2020 grant MSCA-101031081.

\bibliographystyle{alphaurl}
\bibliography{admin/bib.bib}

\end{document}